\documentclass[conference,a4paper]{IEEEtran}

\usepackage{cite}
\usepackage{amsmath,amssymb}
\usepackage{mathtools}
\usepackage{graphicx}
\usepackage{xcolor}
\usepackage{comment}

\newcommand{\medleftparen}{\mathopen{\scalebox{1}[1.15]{\( (\)}}}
\newcommand{\medrightparen}{\mathclose{\scalebox{1}[1.15]{\( )\)}}}

\newtheorem{theorem}{Theorem}

\newtheorem{proposition}{Proposition}
\newtheorem{corollary}{Corollary}
\newtheorem{remark}{Remark}

\title{Sequential Lossy Compression With Causal Conditional Perception}
\author{\IEEEauthorblockN{Photios A. Stavrou and Zixuan He}
\IEEEauthorblockA{Communication Systems Department, EURECOM\\
Sophia-Antipolis, France\\
\texttt{\{fotios.stavrou,zixuan.he\}@eurecom.fr}}}

\begin{document}
\maketitle
\begin{abstract}
In this paper, we study sequential lossy compression under a causal conditional perception criterion comparing source and reconstruction distributions given the same reconstruction history. For first-order Markov sources, we formulate the finite-horizon nonanticipative rate-distortion-perception function (NRDPF) with stagewise constraints and establish one-shot lower and upper bounds on the minimum variable-length sum rate using a strengthened strong functional-representation lemma (SFRL) and common randomness. For time-varying scalar Gauss--Markov sources under pointwise mean-squared error (MSE) and conditional squared Wasserstein-$2$ fidelity, we prove Gaussian optimality, derive a log-variance characterization, and obtain a closed-form solution that recovers the classical Gaussian nonanticipative rate-distortion function (NRDF) when perception is unconstrained and the classical Gaussian RDPF when the source is stationary and memoryless.
\end{abstract}
\section{Introduction}
\label{sec:intro}

Low-latency video compression must exploit temporal dependence while producing perceptually realistic reconstructions at limited rates. Conventional distortion measures, such as the mean-squared error (MSE), quantify samplewise fidelity but may favor blurry or visually implausible outputs at low rates. This limitation motivated the rate-distortion-perception (RDP) framework, in which distortion is complemented by a distributional fidelity constraint \cite{blau:2019}. In sequential video compression, perceptual fidelity must additionally account for the temporal structure of the source and the causal operation of the decoder.
\par Sequential source coding of correlated sources was introduced by  \cite{viswanathan:2000} aiming to provide an information-theoretic framework to study sequential compression problems. This framework was extended by \cite{ma:2011} to characterize the corresponding stagewise rate-distortion region, which retains the complete vector of rates and accounts for possible encoding and decoding delays. A related but distinct problem is to minimize the sum of the rates over the coding horizon. In this direction, \cite{yang:2011} characterized and developed a computational method for the minimum total rate of causal video coding, whereas \cite{stavrou:2021tac2} derived finite-horizon lower and upper bounds, together with dynamic rate-allocation methods, for the minimum total rate of sequentially encoded stochastic systems. Interestingly, the sum-rate formulation in these works results in a lower-bound characterization known as the sequential or nonanticipative rate-distortion function (NRDF) \cite{gorbunov:1973,tatikonda:2004}.

Recently, two companion papers \cite{sadaf:2023,sadaf:2025} incorporated perceptual fidelity into sequential lossy compression. They considered perceptual loss functions (PLFs) on frame-wise marginal distributions \cite{mentzer:2022} and on the joint distribution of the source frames \cite{veerabadran:2020}, and subsequently introduced a self-adaptive perceptual loss function (PLF-SA) that involves the previous reconstructions and the current frame. These works formulate and characterize sequential operational RDP regions. In contrast to \cite{sadaf:2023,sadaf:2025}, \cite{zhang:2026} studies minimum-sum-rates under total-distortion and global Kullback-Leibler perception constraint on the full sequence. We also minimize the sum-rate, extending \cite{sadaf:2023,sadaf:2025} through the \textit{conditional self-adaptive perception loss function (PLF-CSA)}, a fibrewise causal criterion comparing source and reconstruction laws conditioned on the same reconstruction history. This criterion is motivated by predictive video
compression, where the current frame is reconstructed from the received
bitstream using previously reconstructed frames \cite{lu:2019}. It is also relevant to sequential
semantic communication, where the interpretation and relevance
of the current reconstruction may depend on the context
previously recovered by the receiver \cite{strinati:2024}.

{\bf Contributions.} For first-order Markov sources, we
(i) formulate the NRDPF under per-frame expected-distortion
and stagewise conditional-perception constraints, and
(ii) use a strengthened SFRL \cite{li:2026} and common randomness
to derive one-shot lower and upper bounds on the minimum
variable-length prefix-free sum rate. For time-varying scalar
Gauss--Markov sources under MSE and conditional squared
Wasserstein-$2$ fidelity, we (iii) prove Gaussian optimality and
derive a scalar log-variance characterization and closed-form
solution, including perfect realism. The solution recovers the
scalar Gaussian NRDF \cite{gorbunov:1974} when perception is
inactive and the Gaussian RDPF \cite{zhang:2021,serra:2024} for
independent and identically distributed (i.i.d) sources.

\section{Problem statement}\label{sec:prob_stat}

Assume that we have $(n+1)$ video frames, denoted by $X^n=(X_0,\ldots,X_n)\in{\mathcal X^n}=(\mathcal X_0\times\cdots\times\mathcal X_n)$, where $\mathcal X_t\subseteq\mathbb R^p$ for every $t\in\mathbb{N}_0^n=\{0,\ldots,n\}$, distributed according to the joint distribution $P_{X^n}$. We assume that the frame sequence $X^n$ is a first-order Markov source; therefore, its joint distribution factorizes as
\begin{align}
P_{X^n}=P_{X_0}\prod_{t=1}^{n}P_{X_t\mid X_{t-1}}.
\label{first_order_markov}
\end{align}
\textbf{Sequential Operation.} The encoders and decoders have access to shared common randomness $K\in\mathcal K$, which is independent of $X^n$. The possibly stochastic $t$-th encoding function observes the source frames $X^t=(X_0,\ldots,X_t)$ and the common randomness $K$, and outputs a variable-length binary prefix-free message $M_t\in\mathcal M_t$, where $\mathcal M_t=\{0,1\}^{\star}$; that is, 
$f_t:\mathcal X^t\times\mathcal K
\longrightarrow\mathcal M_t,~t\in\mathbb{N}_0^n$.
For every \(t\) and every realization \((m^{t-1},k)\), the conditional codebook used to encode \(M_t\), given \((M^{t-1},K)=(m^{t-1},k)\), is assumed to be prefix-free.
Using $M^t$ and common randomness $K$, the possibly stochastic $t$-th decoding function produces $\hat{X}_t$, that is
$g_t:\mathcal M^t\times\mathcal K
\longrightarrow\mathcal{\hat{X}}_t,~t\in\mathbb{N}_0^n$.
Such a code $\{(f_t,g_t):~t\in\mathbb{N}_0^n\}$ \textit{ideally} induces a causal reproduction distribution
$\prod_{t=0}^n
P_{\hat X_t\mid X^t,\hat X^{t-1}}$.     
We impose a per-frame expected-distortion constraint
\begin{equation}
\mathbb E[d(X_t,\hat X_t)]\le D_t,~D_t\in[D_t^{\min},D^{\max}_{t}]\subseteq(0,\infty],~t\in\mathbb{N}_0^n\label{eq:pointwise-distortion}
\end{equation}
\noindent where by setting $D_t=\infty$ removes the $t$-frame distortion constraint.
To capture causal perceptual fidelity, we compare the source and reconstruction distributions conditioned on the same past reconstructions at each frame. For a divergence ${\mathsf D}$, we define at each $t$ a PLF-CSA as a stagewise conditional perception constrained by perceptual level $P_t\in[0,P^{\max}_{t}]\subseteq[0,\infty)$ as follows
\begin{align} 
\Pi_t
\triangleq
\mathbb E_{\hat X^{t-1}}
\Big[
\mathsf D\big(
P_{X_t\mid\hat X^{t-1}},
P_{\hat X_t\mid\hat X^{t-1}}
\big)
\Big]\le P_t,~t\in\mathbb{N}_0^n.
\label{eq:stagewise_perception}
\end{align}
\par 
\noindent PLF-CSA and PLF-SA coincide for $f$-divergences \cite{sason:2018}, since the joint distributions share the same history marginal, but not in general. For example, assuming squared Wasserstein-$2$ distance, PLF-CSA restricts transport to a common reconstruction history, whereas PLF-SA permits cross-history transport; hence, PLF-SA $\leq$ PLF-CSA and the inequality may be strict. Nonetheless, both criteria have the same perfect-realism condition.\\
\textbf{Operational and information-theoretic sum-rates.} Let \({\bf D}\triangleq(D_0,\allowbreak\ldots,D_n)\) and \({\bf P}\triangleq(P_0,\ldots,P_n)\). The operational one-shot minimum achievable sum-rate under common randomness is defined as 
\begin{align} 
R_{[0,n]}^{\mathrm{op,cr}}({\bf D},{\bf P}) \triangleq \inf_{\{\substack{\{(f_t,g_t)\}_{t=0}^{n}:\ \eqref{eq:pointwise-distortion}, \eqref{eq:stagewise_perception}}\}} \frac{1}{n+1} \sum_{t=0}^{n}\mathbb E[\ell(M_t)]
\label{eq:operational_sum_rate} 
\end{align} 
where the infimum is over all causal conditionally prefix-free codes satisfying \eqref{eq:pointwise-distortion} and \eqref{eq:stagewise_perception}.
\par Let \(\mathcal Q_{[0,n]}({\bf D},{\bf P})\) denote the set of causal reproduction distributions such that
\begin{equation}
\begin{aligned} 
&\mathcal Q_{[0,n]}({\bf D},{\bf P})\triangleq
\Big\{Q_{\hat X^n\Vert X^n} \triangleq \prod_{t=0}^{n} Q_{\hat X_t\mid X^t,\hat X^{t-1}}:\\
&\qquad\mathbb E_Q[d(X_t,\hat X_t)]\le D_t,~~\Pi_t(Q)\le P_t,~ t\in\mathbb{N}_0^n\Big\}.\label{eq:constraint_set} 
\end{aligned} 
\end{equation}
All expectations and distributions in \eqref{eq:constraint_set} are evaluated under the joint distribution \(P_{X^n}Q_{\hat X^n\Vert X^n}\). 
\par We define the NRDPF as 
\begin{align} 
R_{[0,n]}^{\mathrm{na}}({\bf D},{\bf P}) \triangleq \inf_{Q\in\mathcal Q_{[0,n]}({\bf D},{\bf P})} \frac{1}{n+1} I_Q(X^n\rightarrow\hat X^n)\label{eq:nonanticipative_rdp} 
\end{align} 
where $I_Q(X^n\rightarrow\hat X^n)$ is directed information \cite{massey:1990} defined as
\begin{align} 
I_Q(X^n\rightarrow\hat X^n) \triangleq \sum_{t=0}^{n} I_Q(X^t;\hat X_t\mid\hat X^{t-1}). \label{eq:directed_information} \end{align} 
For \(t=0\), the history \(\hat X^{-1}\) is empty, so the corresponding term in \eqref{eq:directed_information} is \(I_Q(X_0;\hat X_0)\).
\begin{remark}[Relation to the NRDF] The quantity in \eqref{eq:nonanticipative_rdp} generalizes the classical nonanticipative, also referred to as sequential, RDF \cite{stavrou:2021tac2,gorbunov:1973,tatikonda:2004} by incorporating stagewise perceptual-fidelity constraints. The classical nonanticipative RDF is recovered when these perception constraints are removed or are inactive. 
\end{remark}
\noindent The following structural property known for first-order Markov sources under single-letter distortion, see e.g. \cite{stavrou:2018siam}, extends directly to the present distribution-dependent causal
perception constraints because Markovization preserves all
stagewise distributions entering the two fidelity criteria.
\begin{proposition}[Markov reproduction distributions]
\label{prop:markov_reproduction_kernel}
Suppose that \(X^n\) is a first-order Markov source subject to the fidelity constraints \eqref{eq:pointwise-distortion} and \eqref{eq:stagewise_perception}. Then,
without loss of optimality, the infimum in
\eqref{eq:nonanticipative_rdp} can be restricted to causal
reproduction distributions of the form
\begin{align}
Q^{\mathrm{MC}}_{\hat X^n\Vert X^n}
=
\prod_{t=0}^{n}
Q^{\mathrm{MC}}_{\hat X_t\mid X_t,\hat X^{t-1}}.
\label{eq:markov_reproduction_kernel}
\end{align}
Consequently, \eqref{eq:nonanticipative_rdp} simplifies to
\begin{align}
R_{[0,n]}^{\mathrm{na}}({\bf D},{\bf P})
= \inf_{\substack{
Q^{\mathrm{MC}}\in\mathcal Q^{\mathrm{MC}}_{[0,n]}({\bf D},{\bf P})}}\frac{1}{n+1}
I_{Q^{\mathrm{MC}}}(X^n\rightarrow\hat X^n)
\label{eq:markov_nonanticipative_rdp}
\end{align}
where $I_{Q^{\mathrm{MC}}}(X^n\rightarrow\hat X^n) \triangleq \sum_{t=0}^{n} I_{Q^{\mathrm{MC}}}(X_t;\hat X_t\mid\hat X^{t-1})$ and
\begin{equation}
\begin{aligned} 
&\mathcal Q^{\mathrm{MC}}_{[0,n]}({\bf D},{\bf P})\triangleq
\Big\{\{Q^{\mathrm{MC}}_{\hat X^n\Vert X^n} \triangleq \prod_{t=0}^{n} Q^{\mathrm{MC}}_{\hat X_t\mid X_t, \hat X^{t-1}}:\\
&\qquad\mathbb E_{Q^{\mathrm{MC}}}[d(X_t,\hat X_t)]\le D_t,~~\Pi_t(Q^{\mathrm{MC}})\le P_t, ~t\in\mathbb{N}_0^n\Big\}.\nonumber
\end{aligned}
\end{equation}
\end{proposition}

\section{One-Shot Achievability}\label{sec:main_results}

In this section, we state our achievability result. 
This relates the operational minimum-sum-rate to the NRDPF. 

\begin{theorem}[Operational bounds for first-order Markov sources]
\label{thm:operational_bounds_markov}
Suppose that \(X^n\) is a first-order Markov source,
$\mathcal Q_{[0,n]}^{\mathrm{MC}}(\mathbf D,\mathbf P)$ is nonempty, and
\(R_{[0,n]}^{\mathrm{na}}(\mathbf D,\mathbf P)<\infty\) in \eqref{eq:markov_nonanticipative_rdp}.
Assume that the encoder and decoder have access to unlimited
common randomness, independent of \(X^n\). Then
\begin{align}
&R_{[0,n]}^{\mathrm{na}}(\mathbf D,\mathbf P)
\le
R_{[0,n]}^{\mathrm{op,cr}}(\mathbf D,\mathbf P)
\label{eq:theorem_converse}\\
&\qquad\quad\le
R_{[0,n]}^{\mathrm{na}}(\mathbf D,\mathbf P)
+\log_2\!\Bigl(
R_{[0,n]}^{\mathrm{na}}(\mathbf D,\mathbf P)+3.4
\Bigr)+2
\label{eq:theorem_logarithmic_gap}
\end{align}
where $R_{[0,n]}^{\mathrm{na}}(\mathbf D, \mathbf P)$ is given by \eqref{eq:markov_nonanticipative_rdp}.
\end{theorem}
\begin{IEEEproof}[Proof sketch]
{\bf Converse.} Consider any causal code satisfying
the distortion and perception constraints. Since the codebook at
time \(t\), conditioned on \((M^{t-1},K)\), is prefix-free, the
source-coding inequality \cite[Ch.~5]{cover-thomas:2006} gives
\begin{align}
\sum_{t=0}^{n}\mathbb E[\ell(M_t)]
&\ge
\sum_{t=0}^{n}H(M_t\mid M^{t-1},K) \label{eq:proof_converse_entropy}\\
&\ge
I(X^n;M^n\mid K)\nonumber\\
&=I(X^n;M^n,K)\label{eq:converse_ineq:1}\\
&\ge I(X^n;\hat X^n)\label{eq:converse_ineq:2}\\
&=I(X^n\rightarrow\hat X^n)\label{eq:converse_ineq:3}
\end{align}
where \eqref{eq:converse_ineq:1} follows because \(K\perp X^n\), \eqref{eq:converse_ineq:2} follows from data processing because the decoder generates \(\hat X^n\) from \((M^n,K)\), and \eqref{eq:converse_ineq:3} follows because source is exogenous, hence $I(\hat X^{t-1};X_t\mid X^{t-1})=0$,~$t\in\mathbb{N}_0^n$ and by the conservation law for directed information
\cite{massey:1995}. \\
The code induces a feasible causal reproduction distribution, and thus
\begin{align*}
I(X^n\rightarrow\hat X^n)
\ge
(n+1)R_{[0,n]}^{\mathrm{na}}({\bf D},{\bf P}).
\end{align*}
Combining this inequality with
\eqref{eq:proof_converse_entropy}-\eqref{eq:converse_ineq:3},
dividing by \(n+1\), and taking the infimum over all feasible
causal prefix-free codes yields \eqref{eq:theorem_converse}
where Proposition \ref{prop:markov_reproduction_kernel} ensures that  $R_{[0,n]}^{\mathrm{na}}({\bf D},{\bf P})$ is given by \eqref{eq:markov_nonanticipative_rdp} and the infimum can be restricted to the conditional distributions of the form in \eqref{eq:markov_reproduction_kernel}.\\
\textbf{Achievability.} For any \(\epsilon>0\), select an \(\epsilon\)-optimal feasible
reconstruction distribution \(Q^{\mathrm{MC}}_\epsilon(\cdot||\cdot)\) satisfying
\begin{align}
\frac{1}{n+1}\sum_{t=0}^{n} I_t(Q^{\mathrm{MC}}_\epsilon)
\le
R_{[0,n]}^{\mathrm{na}}({\bf D},{\bf P})+\epsilon
\label{eq:proof_epsilon_optimal}
\end{align}
where $I_t(Q^{\mathrm{MC}}_\epsilon)
\triangleq
I_{Q^{\mathrm{MC}}_\epsilon}(X_t;\hat X_t\mid\hat X^{t-1}),~\text{for any $t$}$.

At each time \(t\), apply the strengthened conditional SFRL \cite[Theorem 14]{li:2026} to $(X,Y,U)=(X_t,\hat X_t,\hat X^{t-1})$. It provides a random variable \(Z_t\), independent of
\((X_t,\hat X^{t-1})\), and a measurable function \(g_t\) such
that
\begin{align}
\hat X_t
=
g_t(X_t,\hat X^{t-1},Z_t)
\label{eq:proof_sfrl_representation}
\end{align}
has conditional distribution
\(Q^{\mathrm{MC}}_{\epsilon,\hat X_t\mid X_t,\hat X^{t-1}}\), and
\begin{equation}
H(\hat X_t \!\mid\!\hat X^{t-1}\!,Z_t)
\le
I_t(Q^{\mathrm{MC}}_\epsilon)
+\log_2\!\bigl(I_t(Q^{\mathrm{MC}}_\epsilon)+3.4\bigr)+1.
\label{eq:proof_sfrl_entropy}
\end{equation}

Choose \(Z_0,\ldots,Z_n\) successively as fresh random variables
independent of the source and of the previously generated
auxiliary variables, and include \(Z^n\) in the common randomness
\(K\). At time \(t\), both terminals know
\((\hat X^{t-1},Z_t)\). The encoder observes \(X_t\), and computes
\(\hat X_t\) using \eqref{eq:proof_sfrl_representation}, and encodes \(\hat X_t\) into the message $M_t$ using a conditional Shannon code (prefix-free code) to match \(P_{\hat X_t\mid\hat X^{t-1},Z_t}\). Given $(M_t,\hat X^{t-1},Z_t)$, the decoder recovers $\hat X_t$. The Shannon coding bound
\cite[Ch.~5]{cover-thomas:2006} gives
\begin{align}
\mathbb E[\ell(M_t)]
&\le
H(\hat X_t\mid\hat X^{t-1},Z_t)+1 \notag\\
&\stackrel{\eqref{eq:proof_sfrl_entropy}}\le
I_t(Q^{\mathrm{MC}}_\epsilon)
+\log_2\!\bigl(I_t(Q^{\mathrm{MC}}_\epsilon)+3.4\bigr)+2.
\label{eq:proof_stagewise_rate}
\end{align}
The auxiliary variable \(Z_t\) need not be communicated because
it is available to both terminals as common randomness.

A forward induction on \(t\) shows that this construction induces
exactly $P_{X^n}\prod_{t=0}^{n}
Q^{\mathrm{MC}}_{\epsilon,\hat X_t\mid X_t,\hat X^{t-1}}$. Indeed, assuming that the target distribution has been generated up to
time \(t-1\), the representation
\eqref{eq:proof_sfrl_representation} generates \(\hat X_t\)
according to the target conditional distribution. Thus, the complete
joint distribution is preserved. In particular, all distortion
and perception constraints, including the nonlinear conditional
perception constraints, are satisfied exactly.

Taking the infimum and averaging \eqref{eq:proof_stagewise_rate} gives the refined bound
\begin{align}
&R_{[0,n]}^{\mathrm{op,cr}}({\bf D},{\bf P})\le
\inf_{Q^{\mathrm{MC}}\in{\mathcal Q}_{[0,n]}^{\mathrm{MC}}(D,P)}
\frac{1}{n+1}\sum_{t=0}^{n}
\medleftparen I_t(Q^{\mathrm{MC}})\nonumber\\
&\qquad\qquad\qquad\qquad\qquad\qquad+\log_2\!\bigl(I_t(Q^{\mathrm{MC}})+3.4\bigr)+2
\medrightparen.\nonumber
\end{align}
Finally, concavity of the logarithm that allows the use of Jensen's inequality \cite{cover-thomas:2006} implies
\begin{align}
&\frac{1}{n+1}\sum_{t=0}^{n}
\log_2 \medleftparen I_t(Q^{\mathrm{MC}}_\epsilon)+3.4\medrightparen \le
\log_2 \medleftparen
\frac{1}{n+1}\sum_{t=0}^{n}I_t(Q^{\mathrm{MC}}_\epsilon)\nonumber\\
&\qquad\qquad\qquad+3.4\medrightparen \stackrel{\eqref{eq:proof_epsilon_optimal}}\le
\log_2\!\left(
R_{[0,n]}^{\mathrm{na}}(D,P)+\epsilon+3.4
\right).
\end{align}
Substituting this bound into
\eqref{eq:proof_stagewise_rate}, and
letting \(\epsilon\downarrow0\) proves
\eqref{eq:theorem_logarithmic_gap}. 
\end{IEEEproof}
\begin{remark}[Relation to existing achievability bounds]\label{remark:2}
The bound of \cite{li:2026}, used in Theorem~\ref{thm:operational_bounds_markov}, is tighter than standard SFRL of \cite{li:2018} used in \cite{sadaf:2023,sadaf:2025}. Unlike \cite{atay:2026}, which uses binary time-sharing, our construction reproduces the target causal distribution exactly and satisfies all stagewise constraints, eliminating the extra $\frac{1}{n}$ redundancy term.
\end{remark}

\section{Case Study: Jointly Gaussian Processes}
\label{sec:joint_gaussian}
Consider the scalar-valued Gauss--Markov source
\begin{align}
X_{t+1}=\alpha_tX_t+W_t, ~~t\in\mathbb{N}_0^{n-1}
\label{eq:GM}
\end{align}
where $\alpha_t\in\mathbb{R}$ is known, $X_0\sim\mathcal N(0,\sigma^2_{X_0})$ with $\sigma^2_{X_0}>0$, and $W_t\in\mathbb{R}\sim\mathcal N(0,\sigma^2_{W_t})$, is an independent white Gaussian noise process with $\sigma^2_W>0$. We specialize the distortion of \eqref{eq:pointwise-distortion} to squared error, i.e., $d(x_t,\hat{x_t})=({x}_t-\hat{x}_t)^2$, and the perception \eqref{eq:stagewise_perception} to conditional squared Wasserstein distance, i.e., $\mathsf D\big(
P_{X_t\mid\hat X^{t-1}},
P_{\hat X_t\mid\hat X^{t-1}}
\big)=W_2^2\big(
P_{X_t\mid\hat X^{t-1}},
P_{\hat X_t\mid\hat X^{t-1}}
\big)$.
\par In what follows, we establish Gaussian optimality of the Gaussian source-reconstruction processes. 
\begin{proposition}[Gaussian optimality]
\label{prop:joint_gaussianity}
For the source model of \eqref{eq:GM}, the infimum in
\eqref{eq:markov_nonanticipative_rdp} is unchanged when restricted to
source--reconstruction processes that are jointly
Gaussian.
\end{proposition}
\begin{IEEEproof}[Proof sketch]
The standard Gaussian-replacement argument
\cite[Sec.~III-C]{ma:2011} preserves the Gaussian source distribution, the
causal structure, and the MSE distortions, while directed information
cannot increase. It therefore remains only to verify the perception
constraints. Let $H_t=\hat X^{t-1}$ and, for any $\epsilon>0$, select an
$\epsilon$-optimal conditional coupling $(U_t,V_t)\mid H_t$ of
$P_{X_t\mid H_t}$ and $P_{\hat X_t\mid H_t}$. Gaussianizing
$(H_t,U_t,V_t)$ while preserving its mean and covariance gives a
conditional coupling $(U_t^{G},V_t^{G})\mid H_t^{G}$ of the
corresponding Gaussian conditional distributions. Consequently, $\Pi_t^{G}
\leq
\mathbb E\|U_t^{G}-V_t^{G}\|^2=
\mathbb E\|U_t-V_t\|^2
\leq \Pi_t+\epsilon$,
where the equality follows because quadratic costs depend only on
first- and second-order moments. Letting $\epsilon\downarrow0$ gives
$\Pi_t^{G}\leq\Pi_t$ for every $t$. Thus, Gaussian replacement
preserves feasibility and cannot increase the objective.
\end{IEEEproof}

\noindent{\bf Scalar Gaussian innovations parametrization.} In view of Proposition~\ref{prop:joint_gaussianity}, it suffices
to consider jointly Gaussian source-reconstruction processes.
Moreover, by applying a history-dependent translation to the
reconstruction, we may assume without loss of optimality that
\begin{align}
\mathbb E[\hat X_t\mid\hat X^{t-1}]
=
\mathbb E[X_t\mid\hat X^{t-1}],
\qquad t\in\mathbb{N}_0^n.
\label{eq:scalar_conditional_mean_matching}
\end{align}
Indeed, this translation leaves the directed information unchanged
and cannot increase either the MSE or the conditional Wasserstein
perception loss.
Define the common conditional mean
$m_t\triangleq
\mathbb E[X_t\mid\hat X^{t-1}]
=
\mathbb E[\hat X_t\mid\hat X^{t-1}]$, the source prediction innovation $E_t
\triangleq X_t-m_t$, and the reconstruction innovation $\widehat E_t
\triangleq\hat X_t-m_t$. Their variances are denoted by $\lambda_t
\triangleq
\operatorname{var}(E_t)
=\operatorname{var}(X_t\mid\hat X^{t-1}),
\label{eq:scalar_prediction_error_variance} I_t\triangleq
\operatorname{var}(\widehat E_t)
= \operatorname{var}(\hat X_t\mid\hat X^{t-1})$,
and let the filtering error variance be denoted by $\delta_t\triangleq\operatorname{var}(X_t\mid\hat X^t)$. Consequently, $P_{X_t\mid\hat X^{t-1}}
=
\mathcal N(m_t,\lambda_t)$, $P_{\hat X_t\mid\hat X^{t-1}}=
\mathcal N(m_t,I_t)$. Moreover,
every nondegenerate causal Gaussian reproduction distribution $Q^{\mathrm{MC},\rm G}_{\hat{X}_t|X_t,\hat{X}^{t-1}}$ admits the innovations orthogonal realization
\begin{align}
\hat X_t
=
m_t+h_t(X_t-m_t)+V_t,~~ t\in\mathbb{N}_0^n
\label{eq:scalar_innovations_realization}
\end{align}
where $h_t\in\mathbb{R}$ and $V_t\sim\mathcal N(0,\sigma_{V_t}^2)$ is independent of
$(X_t,\hat X^{t-1})$. Regarding \eqref{eq:GM} as the state
equation and \eqref{eq:scalar_innovations_realization} as the
observation equation, the standard Kalman-filter equations
\cite{kailath:2000} give for each $t$
\begin{equation}
I_t=
h_t^2\lambda_t+\sigma_{V_t}^2, \delta_t
=
\lambda_t-\frac{h_t^2\lambda_t^2}{I_t},~\lambda_{t+1}
=
\alpha_t^2\delta_t+\sigma_{W_t}^2.
\label{eq:scalar_kalman_filter}
\end{equation}
Equivalently, selecting the nonnegative correlation between the
source and reconstruction innovations, the realization parameters
can be expressed as
\begin{align}
h_t=
\frac{\sqrt{I_t(\lambda_t-\delta_t)}}{\lambda_t},~\sigma_{V_t}^2=
\frac{I_t\delta_t}{\lambda_t}.
\label{eq:scalar_realization_parameters}
\end{align}
The degenerate case $I_t=0$ is obtained separately and necessarily
corresponds to $\delta_t=\lambda_t$.
\noindent Using Proposition \ref{prop:joint_gaussianity}, and the Gaussian innovations parametrization above, we obtain the following characterization of the scalar-valued Gaussian NRDPF.
\begin{theorem}[Scalar Gaussian characterization of the NRDPF]
\label{thm:scalar_gaussian_characterization}
Consider the time-varying scalar Gauss--Markov source in \eqref{eq:GM}
when \eqref{eq:pointwise-distortion} is the pointwise MSE distortion and when \eqref{eq:stagewise_perception} is the causal conditional squared
Wasserstein-$2$ perception criterion. Then for  $\mathbf D\in(0,\infty]^{n+1}$ and $\mathbf P\in[0,\infty)^{n+1}$, the Gaussian NRDPF is characterized by
\begin{align}
&R_{[0,n]}^{\mathrm{na}}
(\mathbf D,\mathbf P)
=
R_{[0,n]}^{\mathrm{na,G}}
(\mathbf D,\mathbf P) 
\nonumber\\
&\qquad=\inf_{\substack{
\lambda_t>0,\;
0<\delta_t\leq\lambda_t,\;
I_t\geq0\\
\delta_t+\left(\sqrt{\lambda_t-\delta_t}-\sqrt{I_t}
\right)^2\leq D_t\\
\left(
\sqrt{\lambda_t}
-
\sqrt{I_t}
\right)^2\leq
P_t, \
t\in\mathbb{N}_0^n
}}
\frac{1}{2(n+1)}
\sum_{t=0}^{n}
\log_2
\left(
\frac{\lambda_t}{\delta_t}
\right)
\label{eq:scalar_gaussian_nrdpf_characterization}
\end{align}
with $\lambda_{t+1}=\alpha_t^2\delta_t+\sigma_{W_t}^2$ and $\lambda_0=\sigma_{X_0}^2$.
If $I_t=0$, feasibility is understood to require
$\delta_t=\lambda_t$. Moreover, the characterization of \eqref{eq:scalar_gaussian_nrdpf_characterization} is achieved by the orthogonal realization in \eqref{eq:scalar_innovations_realization}.
\end{theorem}
\begin{IEEEproof}[Proof sketch]
From Proposition \ref{prop:joint_gaussianity} and the innovations parametrization above, the objective in \eqref{eq:markov_nonanticipative_rdp} yields \eqref{eq:scalar_gaussian_nrdpf_characterization}. 
Moreover, from \eqref{eq:GM} we obtain $\lambda_{t+1}=\alpha_t^2\delta_t+\sigma_{W_t}^2$ in \eqref{eq:scalar_kalman_filter}, with
$\lambda_0=\sigma_{X_0}^2$. Let $C_t$ denote the covariance between the source and reconstruction
innovations. For $I_t>0$, Gaussian conditioning gives
$\delta_t=\lambda_t-C_t^2/I_t$. Since changing the sign of $C_t$
does not affect the rate or perception loss, while the MSE is
minimized by $C_t\geq0$, we may take
$C_t=\sqrt{I_t(\lambda_t-\delta_t)}$. Conditional mean matching in \eqref{eq:scalar_conditional_mean_matching} and
the scalar Gaussian Wasserstein formula \cite[Eq.~(1.6), p.~18]{panaretos:2020} yield
$\mathbb E[(X_t-\hat X_t)^2]
=
\delta_t+
\left(
\sqrt{\lambda_t-\delta_t}-\sqrt{I_t}
\right)^2,~\Pi_t
=
\left(
\sqrt{\lambda_t}-\sqrt{I_t}
\right)^2.$ 
Hence, every feasible Gaussian reproduction distribution induces a
feasible tuple in
\eqref{eq:scalar_gaussian_nrdpf_characterization} with the same rate. Conversely, any feasible nondegenerate tuple is realized by
\eqref{eq:scalar_innovations_realization} with design variables \eqref{eq:scalar_realization_parameters}.
These choices induce the prescribed values of
$(\lambda_t,\delta_t,I_t)$ and therefore attain the objective and
both fidelity constraints. If $I_t=0$, the reconstruction innovation
vanishes and necessarily $\delta_t=\lambda_t$ (zero rate), giving the corresponding degenerate realization.
\end{IEEEproof}

Next, we give the solution of  \eqref{eq:scalar_gaussian_nrdpf_characterization}.
\begin{theorem}[Closed-form solution]
\label{thm:closed_form}
Let $\mathbf D\in(0,\infty]^{n+1}$ and $\mathbf P\in[0,\infty)^{n+1}$. Then, the infimum in
\eqref{eq:scalar_gaussian_nrdpf_characterization} is attained, and its
optimal value is
\begin{align}
R_{[0,n]}^{\mathrm{na,G}}
(\mathbf D,\mathbf P)=
\frac{1}{2(n+1)}
\sum_{t=0}^{n}
\log_2
\left(
\frac{\lambda_t^\star}
     {\delta_t^\star}
\right)
\label{eq:scalar_gaussian_nrdpf}
\end{align}
where $\lambda^*_{t+1}=\alpha^2_t\delta^*_{t}+\sigma^2_{W_t}$, $\lambda^*_0=\sigma^2_{X_0}$ and
\begin{align}
\delta_t^\star
=
\begin{cases}
\min\{\lambda_t^\star,D_t\},
&
\ell_t^\star
\leq
\sqrt{|\lambda_t^\star-D_t|}
\\[1.5ex]
\displaystyle
\lambda_t^\star
-
\frac{
\left(
(\ell_t^\star)^2+\lambda_t^\star-D_t
\right)^2
}{
4(\ell_t^\star)^2
},
&
\ell_t^\star
>
\sqrt{|\lambda_t^\star-D_t|}
\end{cases}
\label{eq:scalar_optimal_filtering_variance}
\end{align}
with 
\begin{align}
\ell_t^\star
\triangleq
\left(
\sqrt{\lambda_t^\star}-\sqrt{P_t}
\right)_{+}
\label{eq:scalar_stagewise_ell}
\end{align}
such that $(x)_{+}\triangleq\max\{x,0\}$. 
\end{theorem}
\begin{IEEEproof}[Proof sketch]
Nonemptiness and finiteness follow by successively choosing
$I_t=\lambda_t$ and
$\delta_t\in(0,\min\{\lambda_t,D_t/2\}]$. This gives zero perception
loss and distortion at most $2\delta_t\leq D_t$, while the resulting
objective is finite. For fixed $(\lambda_t,\delta_t)$, the objective
does not depend on $I_t$. Hence, projecting
$\sqrt{\lambda_t-\delta_t}$ onto the perception-feasible interval
$\ell_t\leq\sqrt{I_t}\leq\sqrt{\lambda_t}+\sqrt{P_t}$ gives
$I_t^\star
=
\left[
\max\left\{
\sqrt{\lambda_t-\delta_t},\ell_t
\right\}
\right]^2$
and the two fidelity constraints reduce to
$\phi_t(\lambda_t,\delta_t)
\triangleq
\delta_t+
\left[
\ell_t-\sqrt{\lambda_t-\delta_t}
\right]_+^2
\leq D_t$.
Since $\phi_t(\lambda_t,\cdot)$ is increasing whereas the stage-$t$
rate is decreasing in $\delta_t$, the optimal
$\delta_t^\star=F_t(\lambda_t)$ is the largest feasible value in
$(0,\lambda_t]$. This is the endpoint
$\min\{\lambda_t,D_t\}$ whenever
$\ell_t\leq\sqrt{|\lambda_t-D_t|}$, giving the first branch of
\eqref{eq:scalar_optimal_filtering_variance}. Otherwise, both fidelity
constraints are active, and solving
$\ell_t=\sqrt{\lambda_t-\delta_t}+\sqrt{D_t-\delta_t}$
by a difference-of-squares argument yields the second branch of \eqref{eq:scalar_optimal_filtering_variance}. It remains to verify global optimality under
$\lambda_{t+1}=\alpha_t^2\delta_t+\sigma_{W_t}^2$. For $t<n$, the
paired term
$-\log\delta_t+\log\lambda_{t+1}$ is strictly decreasing in
$\delta_t$, while the terminal term $-\log\delta_n$ is strictly
decreasing in $\delta_n$. Moreover,
$\phi_{t+1}(\lambda,\delta)$ is nonincreasing in $\lambda$ for fixed
$\delta$; hence, increasing $\delta_t$, and therefore
$\lambda_{t+1}$, can only relax the next-stage feasibility constraint.
Thus, the forward recursions
$\delta_t^\star=F_t(\lambda_t^\star)$ and
$\lambda_{t+1}^\star
=\alpha_t^2\delta_t^\star+\sigma_{W_t}^2$
are globally optimal. Together with the corresponding $I_t^\star$,
they attain the infimum and yield
\eqref{eq:scalar_gaussian_nrdpf}.
\end{IEEEproof}
\begin{figure}
    \centering
    \includegraphics[width=\linewidth]{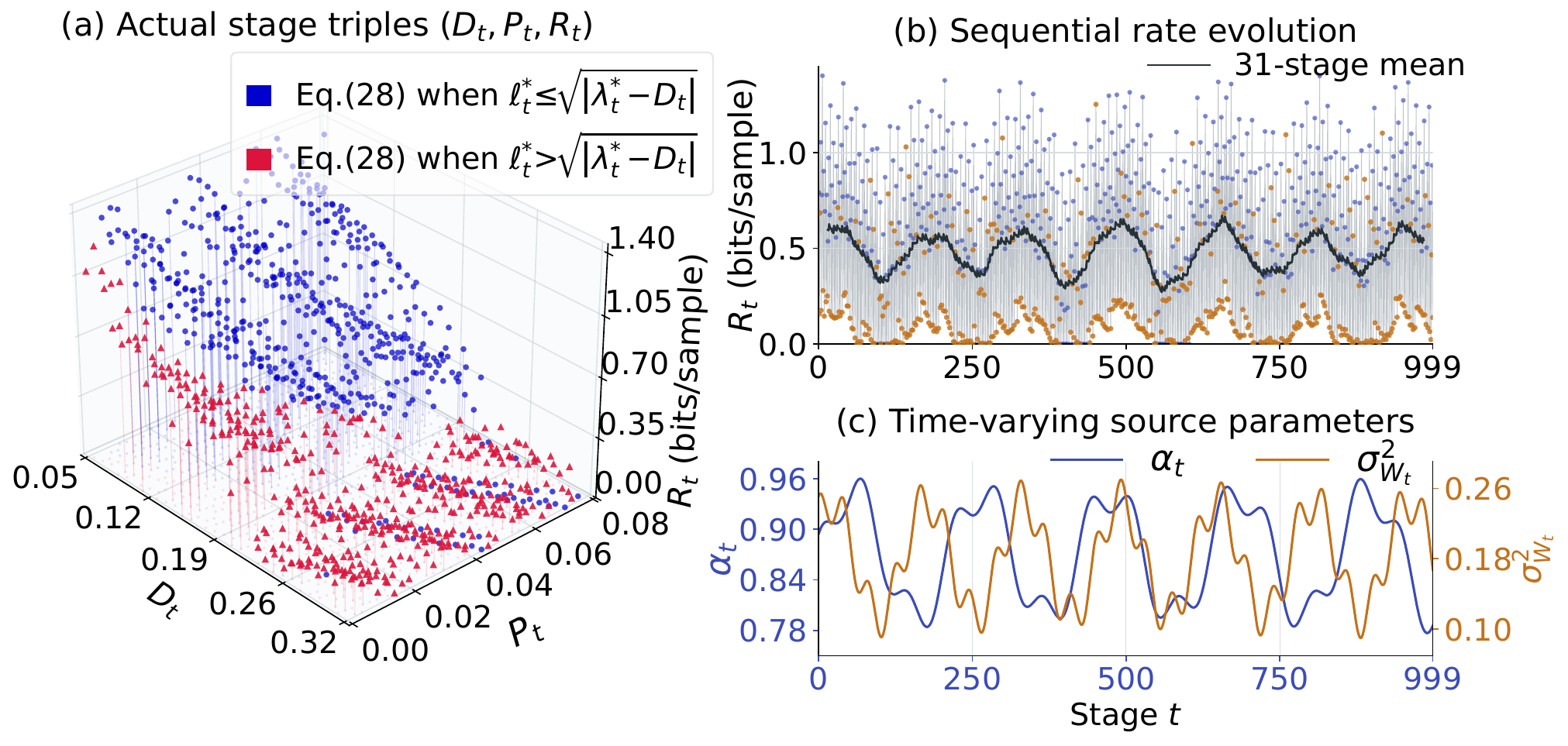}
    \caption{Stagewise reverse-waterfilling solution over \(1000\) stages.}
    \label{figure:1}
\end{figure}
\noindent\textbf{Numerical example.}
In Fig. \ref{figure:1}, we give a pictorial view of the behavior of Theorem \ref{thm:closed_form} for a specific example. We consider $1000$ stages $t=0,\ldots,999$ with $\lambda_0=0.4$, $\alpha_t\in[0.777,0.960]$, and
$\sigma_{W_t}^2\in[0.090,0.270]$. The fidelity budgets are generated using $D_t\in[0.05,0.32]$ and $P_t\in[0,0.08]$. At each stage, Theorem~\ref{thm:closed_form} is applied recursively to compute $\delta_t^\star, \lambda_{t+1}^\star$, and
$R_t=\frac12\log_2(\frac{\lambda_t^\star}{\delta_t^\star})$. Among the $1000$ stages, $526$ employ the MSE-only solution, whereas $474$ employ the MSE-perception solution. \par Theorem~\ref{thm:closed_form} reveals a two-branch, stagewise reverse-waterfilling structure. The MSE constraint determines the classical allocation $\min\{\lambda^\star_t,D_t\}$ while the perception constraint introduces the stage-dependent water level $\ell_t^*$. When this level exceeds the MSE threshold, the solution shifts to the joint MSE–perception branch of \eqref{eq:scalar_optimal_filtering_variance}.
\begin{corollary}[Perfect-realism solution]
\label{cor:perfect_realism}
Suppose that in Theorem~\ref{thm:closed_form}, we let
$\mathbf P=\mathbf 0$ (zero vector). Then
\begin{align}
R_{[0,n]}^{\mathrm{na,G}}(\mathbf D,\mathbf 0)
=
\frac{1}{2(n+1)}
\sum_{t=0}^{n}
\log_2
\left(
\frac{\lambda_t^\star}{\delta_t^\star}
\right)
\end{align}
where
\begin{align}
\delta_t^\star
=
\begin{cases}
\displaystyle
D_t-\frac{D_t^2}{4\lambda_t^\star},
&0<D_t<2\lambda_t^\star\\[2mm]
\lambda_t^\star,
&D_t\geq2\lambda_t^\star.
\end{cases}
\end{align}
\end{corollary}
\begin{IEEEproof}
This follows from Theorem \ref{thm:closed_form}.
\end{IEEEproof}
\begin{remark}[Special cases]\label{rem:special_cases} 
{\bf (i)} When the perception constraint is inactive, Theorem~\ref{thm:closed_form}
reduces to the classical scalar Gaussian NRDF \cite{gorbunov:1974}, with $\delta_t^\star=\min\{\lambda_t^\star,D_t\}$ and
$I_t^\star=\lambda_t^\star-\delta_t^\star$. {\bf (ii)} If in \eqref{eq:GM} we set $\alpha_t=0$ and $\sigma_{X_0}^2=\sigma_{W_t}^2=\sigma_X^2$ for all $t$, the source is i.i.d. Gaussian. In this case, the optimal
reproduction distribution is memoryless, and hence $P_{X_t\mid\hat X^{t-1}}=P_{X_t}$, $P_{\hat X_t\mid\hat X^{t-1}}=P_{\hat X_t}$.
Consequently, the conditional perception criterion reduces to
$W_2^2(P_{X_t},P_{\hat X_t})$, and our characterization recovers 
the classical Gaussian RDPF of~\cite{zhang:2021} (see also \cite{serra:2024}). 
{\bf (iii)} At perfect
realism, when the source is i.i.d., Corollary~\ref{cor:perfect_realism} recovers the distribution-preserving RDF in \cite[Proposition 2]{li:2011}.   
\end{remark}
\section{Conclusion}
We studied sequential lossy compression under a causal conditional
perception criterion. For first-order Markov sources, we formulated
the NRDPF and established one-shot lower and upper
bounds on the operational variable-length sum-rate using
shared common randomness and a strengthened SFRL. For time-varying
scalar Gauss--Markov sources under MSE and conditional squared
Wasserstein-$2$ fidelity, we proved Gaussian optimality and derived an
explicit characterization and its closed-form solution,
including the perfect-realism regime. Extensions to vector sources
and practical predictive video-coding architectures are left for
future work.

\section*{Acknowledgment}
The work of P. A. Stavrou and Z. He was partially supported by a Huawei France-EURECOM Chair on Future Wireless Networks. The work of P. A. Stavrou is also supported by the SNS JU project 6G-GOALS under the EU Horizon programme (Grant Agreement No. 101139232).

\bibliographystyle{IEEEtran}
\bibliography{string,refs}

\end{document}